\documentclass[12pt]{article}

\usepackage{epsfig}
\usepackage{amssymb}
\usepackage{amsfonts}

\usepackage{color}

\def\be{\begin{equation}}
\def\ee{\end{equation}}
\def\ba{\begin{array}{c}}
\def\ea{\end{array}}

\newcommand{\bea}{\begin{eqnarray}}
\newcommand{\eea}{\end{eqnarray}}

\newtheorem{thm}{Theorem}

\newtheorem{lemma}[thm]{Lemma}

\newenvironment{proof}{\noindent {\bf Proof}}{\hfill$\square$\vspace{3mm}\endtrivlist}

\begin{document}

 \begin{center}{\Large \bf

Non-Hermitian Bose-Hubbard-like quantum models

  }\end{center}


 \begin{center}

\vspace{8mm}

  {\bf Miloslav Znojil} $^{1,2,3,4}$

\end{center}

\vspace{8mm}

  $^{1}$
 {The Czech Academy of Sciences,
 Nuclear Physics Institute,
 Hlavn\'{\i} 130,
250 68 \v{R}e\v{z}, Czech Republic, {e-mail: znojil@ujf.cas.cz}}


 $^{2}$
 {Department of Physics, Faculty of
Science, University of Hradec Kr\'{a}lov\'{e}, Rokitansk\'{e}ho 62,
50003 Hradec Kr\'{a}lov\'{e},
 Czech Republic}

%
%
%
%
%
%

{$^3$Institute of System Science, Durban University of
Technology, Durban,
 South Africa}


{$^4$
School for Data Science and Computational Thinking, Stellenbosch
University, 7600 Stellenbosch,
 South Africa}

\section*{Abstract}

{}

A family of multibosonic
complex-symmetric Hamiltonians
possessing both the real and complex spectra
is studied, with emphasis upon the
properties of the latter subfamily.
In it one treats
resonances as eigenstates of a non-Hermitian effective quantum Hamiltonian.
As long as
the search for their complex energy eigenvalues is not easy,
a reduced task is considered in which one only evaluates the
auxiliary real quantities called singular
values. Several forms of representation of Green's
functions in terms of (possibly, matrix) continued fractions
are shown to offer an efficient approach to this task.
%
%
%
%
%
%
%
%
%


\section{Bose-Hubbard-like Hamiltonians.\\.}

The standard three-parametric many-body Bose-Hubbard quantum Hamiltonian \cite{BHH,Graefe}
 \be
  \label{Ham1}
 H^{(BH)}(\varepsilon,v,c) = \varepsilon\left(a_1^{\dagger}a_1
 - a_2^{\dagger}a_2\right) +
  v\left(a_1^{\dagger}a_2 + a_2^{\dagger}a_1\right) + \frac{c}{2}
  \left( a_1^{\dagger}a_1 - a_2^{\dagger}a_2\right)^2
 \ee
attracts attention due to its combination of realistic aspects
({\it pars pro toto\,} let us mention that it
is of relevance for the study of Bose-Einstein condensation
processes \cite{BEC,BECb,toba3})
with formal merits (due to its numerous symmetries,
the model can be
considered, in some sense \cite{Zhang}, exactly solvable).
Interested readers may find a more extensive introductory outline of both of
these features
of the model
in~\cite{Graefe}.

\subsection{Unconventional non-Hermitian model.\\.}

A key message as delivered in the latter paper is that the authors
proposed to
replace model (\ref{Ham1})
by its
analytically continued
and manifestly non-self-adjoint unconventional Bose-Hubbard alternative,
 $$
 H_{(UBH)}(\gamma,v,c)=
 H^{(BH)}({\rm i}\gamma,v,c)
 \,. $$
In their analysis of the new model they
made ample use of analytic, numerical and
perturbation-expansion methods,
with the results which admitted
the physical interpretation
of the energy spectra in
the framework of the so called open-system theory \cite{Nimrod}
(in this setting the spectrum need not be all real)
as well as  in
the framework of the so called closed-system theory \cite{Geyer}
(in this case the spectrum {\em must\,} be all real \cite{toba1}).

In our present paper we will feel slightly more inclined to prefer the
latter framework but
as long as
our present study and
results were mainly inspired by paper \cite{Graefe},
our preferences will not be strict \cite{toba4}.

\subsection{Example.\\.}

In the model  $H^{(BH)}({\rm
i}\gamma,v,c)$ of Eq.~(\ref{Ham1}) and of Ref. \cite{Graefe}
we will mainly work with the special cases in which the
interaction between particles
(measured by the coupling constant $c$) vanishes.
For the sake of simplicity we will keep
the number of particles
${\cal N}$
unchanged (cf. \cite{toba2})
and we will fix the tunneling
sgtrength $v=1$.
Our
Hamiltonians then become represented
by the $K$ by $K$ complex tridiagonal matrices with
$K={\cal N}+1$, having the form
 \be
 H^{(2)}_{(UBH)}(\gamma)=
 \left[ \begin {array}{cc} -i{\it \gamma}&1
 \\\noalign{\medskip}1&i{\it
 \gamma}
 \end {array} \right]\,
 \label{dopp2}
 \ee
at ${\cal N}=1$, etc.
Closed formulae for the bound-state
energies will be then also available
yielding,
in the single-particle case,
elementary doublet $E_\pm = \pm \sqrt{1-\gamma^2}$, etc
(see also a few further
samples of the
explicit formulae for the ${\cal N}>1$  eigenvalues in \cite{afew}).


\section{Physical appeal of the models.\\.}

\subsection{Conventional non-Hermitian Bose-Hubbard system.\\.}

One of the immediate consequences of the exact
solvability of the above-mentioned  ${\cal N}=1$ model
is that we can easily prove that
the Hamiltonian ceases to be diagonalizable
in the limit of $\gamma \to \pm 1$.
Beyond these boundaries
known, to mathematicians, as ``exceptional points'' (EPs, \cite{Kato})
the spectrum ceases to be real,
i.e., it can only be given
an open-system interpretation.
Inside the open interval of $\gamma \in (-1,1)$,
the system can be interpreted as unitary
(see, e.g., reviews \cite{Geyer,Carl} for
a detailed explanation).
At the two EP boundaries themselves,
the model cannot be assigned any physical
meaning because the Hamiltonian ceases
to be diagonalizable \cite{ali,book}.

Explicit
proofs become particularly easy at ${\cal N}=1$
but they also remain feasible at ${\cal N}=2=K-1$, with
 \be
H^{(3)}_{(UBH)}(\gamma)=\left[ \begin {array}{ccc} -2\,i\gamma&
\sqrt{2}&0\\\noalign{\medskip}\sqrt{2}&0&
\sqrt{2}\\\noalign{\medskip}0&\sqrt{2}&2\,i\gamma\end {array}
\right]\,
  \label{3wg}
 \ee
(here we have $E_0=0$ and $E_\pm = \pm 2\,\sqrt{1-\gamma^2}$
and $\gamma^{(EP)}_\pm
=\pm 1$), etc.

\subsection{Generalized non-Hermitian Bose-Hubbard-like models.\\.}

With the growth of  ${\cal N}$ and $K$,
the related mathematics becomes less and less straightforward.
At the same time,
the
existence of the
symmetries behind the
non-Hermitian Bose-Hubbard Hamiltonians
still keeps the models solvable.
At the same time,
some of the related
phenomenological needs
(like, e.g., the
requirement of the existence of exceptional points)
appeared to be satisfied even
after a loss of some of the underlying Lie-algebraic symmetries.

For illustration
let us pick up the four-particle model
 \be
 H^{(5)}_{(UBH)}(\gamma)=\left[
 \begin {array}{ccccc} -4\,i\gamma&2&0&0&0
 \\
 \noalign{\medskip}2&-2\,i\gamma&\sqrt {6}&0&0
 \\
 \noalign{\medskip}0&\sqrt{6}&0&\sqrt {6}&0\\\noalign{\medskip}0&0&
 \sqrt {6}&2\,i\gamma&2\\\noalign{\medskip}0&0&0 &2&4 \,i\gamma\end
 {array} \right]\,
 \label{petpa}
 \ee
in which one still detects the presence of the
two exceptional points $\gamma^{EP}$ of
maximal order equal to $K=5$.

In the context of this particular and already less trivial
complex Bose Hubbard model
we found it challenging to
search, in 2018, for another, non-equivalent four-particle model
exhibiting the same degree of the
non-Hermitian EP-representing degeneracy
even after a weakening of the model's symmetries.
Not quite expectedly, we succeeded \cite{4a5,4a5b}.
Using some
brute-force direct-construction methods
we managed to construct
 \be
 H^{(5)}_{(non-BH)}(\gamma)=\left[
 \begin {array}{ccccc} -4\,i\gamma&8&0&0&0
 \\
 \noalign{\medskip}8&-2\,i\gamma&i\sqrt {54}&0&0
 \\
 \noalign{\medskip}0&i\sqrt{54}&0&i\sqrt {54}&0\\\noalign{\medskip}0&0&
 i\sqrt {54}&2\,i\gamma&8\\\noalign{\medskip}0&0&0 &8&4 \,i\gamma\end
 {array} \right]\,
 \label{zpetpa}
 \ee
i.e., a new complex-symmetric model
of a non-Bose-Hubbard {\it alias\,} Bose-Hubbard-like type.
In this manner, the questions connected with the
evaluation of the spectrum
entered the scene.
As a consequence, our attention
has been redirected to the numerical techniques
and, in particular, to the possibilities of the localization of
the energies via the poles of Green's functions \cite{MCF1}.


\section{Green's functions.\\.}

\subsection{Tridiagonal Hamiltonians and analytic continued fractions.\\.}

Among the most user-friendly complex symmetric matrix models
we have to mention, in the context of numerical methods,
the fairly general class of tridiagonal non-Hermitian quantum
Hamiltonians
 \be
 H=
  \left[ \begin {array}{ccccc}
     a_1&b_1&0
 &0&\ldots
   \\
   c_2&a_2&b_2&0&\ddots
   \\
 0
 &c_3&a_3&b_3&\ddots
   \\
 \vdots&\ddots&\ddots&\ddots&\ddots
    \\
 \end {array} \right]\ \
 \neq
 H^\dagger =
  \left[ \begin {array}{ccccc}
     a_1^*&c_2^*&0
 &0&\ldots
   \\
   b_1^*&a_2^*&c_3^*&0&\ddots
   \\
 0
 &b_2^*&a_3^*&c_4^*&\ddots
   \\
 \vdots&\ddots&\ddots&\ddots&\ddots
    \\
 \end {array} \right]\,,
 \label{trufinkit}
 \ee
the most important mathematical property of which
is the possibility of expressing their
analytic Green's functions
in terms of continued fractions,
 \be
 {\mathcal G}(z)=\frac{1}{
 a_1-z-\frac{b_1c_2}{a_2-z-\frac{b_2c_3}{a_3-z-\ldots}} }\,.
 \label{Virendra}
 \ee
The approach
found a strong motivation in quantum physics
by its property of offering, in many models, a {convergent resummation} of
perturbation series \cite{Singh}. In the language of mathematics,
the basic ideas of such an approach are straightforward:
The tridiagonal eigenvalue problem
 \be
  \left( \begin {array}{ccccc}
     a_1&b_1&0
 &\ldots&0
   \\
   c_2&a_2&b_2&\ddots
 &\vdots
   \\
 0
 &\ddots&\ddots&\ddots&0
   \\
 \vdots&\ddots&c_{N-1}&a_{N-1}&b_{N-1}
    \\
  0&\ldots&0&c_{N}&a_{N}
    \\
 \end {array} \right)\, \left( \begin {array}{c}
 \psi_1\\
 \psi_2\\
 \vdots\\
 \psi_N\\
 \end {array} \right)
 =
 E
 \,\left( \begin {array}{c}
 \psi_1\\
 \psi_2\\
 \vdots\\
 \psi_N\\
 \end {array} \right)\,,
 \ \ \ \ N \to \infty\,
 \label{SEfinkit}
 \ee
is simply reformulated using factorization ansatz
 \be
 H-E
 = {\cal U}\,
 {\cal F}\,
 {\cal L}\,.
 \label{finkit}
 \ee
Here, it makes sense to require that $\det  {\cal U}= \det  {\cal L}=1$,
and that
the central factor matrix ${\cal F}$ remains diagonal, with elements
 $$
 1/f_1,  1/f_2, \ldots,  1/f_N\,.
 $$
This enables one to abbreviate ${u}_{k+1}=-b_kf_{k+1}$
and ${v}_j=-f_jc_j$ and define
the other two matrix factors in the respective bidiagonal forms
 \be
 {\cal U}=
  \left[ \begin {array}{ccccc}
  1&-u_2&0
 &\ldots&0
   \\
     0&1&-u_3&\ddots
 &\vdots
   \\
 0
 &0&\ddots&\ddots&0
   \\
 \vdots&\ddots&\ddots&1&-u_N
    \\
  0&\ldots&0&0&1
    \\
 \end {array} \right]\,,
 \ \ \ \ \ \
 {\cal L}=
  \left[ \begin {array}{ccccc}
  1&0&0
 &\ldots&0
   \\
     -v_2&1&0
 &\ldots&0
   \\
   0&-v_3&\ddots&\ddots
 &\vdots
   \\
 \vdots&\ddots
 &\ddots&1&0
   \\
  0&\ldots&0&-v_{N}&1
    \\
 \end {array} \right]\,.
 \label{lowkit}
 \ee
Non-linear continued-fraction recurrences
 \be
 f_k=\frac{1}{a_k-E-b_kf_{k+1}c_{k+1}}\,,\ \ \ \
 k=N, N-1,\ldots,2 ,1 \,
 \label{cf}
 \ee
now emerge
as a condition of the validity of the factorization ansatz.
At any finite matrix dimension $N$ it is merely sufficient
to put
$f_{N+1}=0$ and {$1/f_1=0$}.

\subsection{Proofs of convergence using the fixed-point-attraction technique.\\.}

In many realistic models with $N \gg 1$ or  $N=\infty$ (cf., e.g., \cite{Singh})
one needs to prove the existence of the Green's-function-representing
continued-fraction expression {$1/f_1=1/f_1(z)$} (cf. also Eq.~(\ref{Virendra}) above).
For our present purposes
we propose to prefer the  fixed-point philosophy of
Ref.~\cite{FP}.
In order to remind the readers about the idea, let us
consider just an instructive
complex  symmetric
$N=\infty$  model
 \be
 H=
  \left[ \begin {array}{cccc}
     \beta_1+i\gamma_1&\alpha_1&0
 &\ldots
   \\
   \alpha_1&\beta_2+i\gamma_2&\alpha_2&\ddots
   \\
 0
 &\alpha_2&\beta_3+i\gamma_3&\ddots
   \\
 \vdots&\ddots&\ddots&\ddots
    \\
 \end {array} \right]\,
 \label{vfinkit}
 \ee
in which we set, for the sake of brevity, $\gamma_k =
0$. Moreover, we will assume that
in the continued-fraction recurrences (\ref{cf})
the $k-$dependence can be scaled out
at $k \gg 1$
so that we can only consider their simplified,
$k-$independent form
 \be
 f'=1/(\beta-E-\alpha^2f)\,.
 \label{fpcf}
 \ee
After another re-scaling we may fix
$\alpha = 1/\sqrt{2}$, set
$\beta=\sqrt{2}(1+\delta^2)$, and neglect $E$.

In the resulting trivial recurrence
 \be
 f'=\frac{2}{2\beta-f}\,
 \label{ipro}
 \ee
one reveals the existence of two fixed points,
 \be
 f_{FP}^{(\pm)}=\beta\pm \sqrt{\beta^2-2}\,.
 \label{tyfps}
 \ee
This means that
$[f_{FP}^{(\pm)}]^2/2=1 \pm \sqrt{2}|\delta| +{\cal O}(\delta^2)$
so that  $f \to f_{FP}^{(-)}=1/\beta+
corrections\ $ since
 \be
 \left .\frac{\partial f'}{\partial f}\right |_{f=f_{FP}^{}}
 =\frac{1}{2}\,
 \left (f_{FP}^{}\right )^2 \,,
 \ \ \ \
 \left |\frac{\partial f'}{\partial f}\right | < 1\,.
 \ee
From these formulae, the proof or disproof of convergence immediately follows.

\section{Singular values.\\.}

Once one admits, in the next step, that $\gamma_k \neq
0$, the simplicity of the proof gets lost.
There are two ways out of the dead end. Either one
gives up the simplicity of the proof,
or one gets inspired by the Pushnitzky's and \v{S}tampach's
abstract mathematical ideas \cite{PS} and
turns attention from the user-unfriendly complex eigenvalues
to the techniques of evaluation of the
real alternative
characteristics of the system called singular values.

 \subsection{The trick of Hermitization.\\.}

Singular values $\sigma_n$ of an operator  $H$
(i.e., say, of our quantum Hamiltonian (\ref{trufinkit}))
can be defined as the non-negative
square roots of eigenvalues $\sigma_n^2$ of the operator product
$H^\dagger H$ or, equivalently, of its
``Hermitized'' partitioned descendant
 \be
  \widetilde{\mathbb H}=
  \left (\begin{array}{c|c}
  0&H\\
  \hline
  H^\dagger&0
  \end{array}
  \right )\,.
  \label{desce}
  \ee
In a way recommended by Pushnitzky with \v{S}tampach
\cite{PS}, a decisive simplification
of the evaluation of the singular values of
our tridiagonal Hamiltonians (\ref{trufinkit})
can be achieved by a suitable
permutation of the basis
using an intertwining relation
 \be
 \widetilde{\mathbb H}\,\mathbb{V}=\mathbb{V}\,\mathbb{H}
 \ee
where
 \be
 \mathbb{V}=
  \left[ \begin {array}{cccccccc}
    1&0&0&0&\ldots&&&
   \\
    0&0&1&0&0&\ldots&&
   \\
    0&0&0&0&1&0&\ldots&
   \\
    \vdots&\ddots&&&\ddots&\ddots&\ddots&
   \\
   \hline
    0& 1&0&0&0&\ldots&&
   \\
    0&0&0&1&0&0&\ldots&
   \\
   0& 0&0&0&0&1&0&\ldots
   \\
   \vdots& \ddots&&&&\ddots&\ddots&\ddots
 \end {array} \right]\,.
 \label{triceps}
  \ee
This leads to the
isospectral partner
${\mathbb H}$
of matrix (\ref{desce})
which is {block-tridiagonal},
 \be
 {\mathbb H}=
  \left[ \begin {array}{cccc}
     A_1&B_1&0
 &\ldots
   \\
   C_2&A_2&B_2&\ddots
   \\
 0
 &C_3&A_3&\ddots
   \\
 \vdots&\ddots&\ddots&\ddots
    \\
 \end {array} \right]\,.
 \label{matrix}
 \ee
Another,
serendipitious advantage
of the permutation is that
the new two-by-two submatrices
of
${\mathbb H}$
are all {sparse},
 \be
 A_k=
\left (
\begin{array}{cc}
0&a_k\\
a_k^*&0 \ea \right )\,,\ \ \ \ \ B_k= \left (
\begin{array}{cc}
0&b_k\\
c_{k+1}^*&0 \ea \right )\,,\ \ \ \ \ C_{k+1}=B_k^\dagger= \left (
\begin{array}{cc}
0&c_{k+1}\\
b_k^*&0 \ea \right )\,.
 \label{[14]}
 \ee
Now, it is
easy to see that
in full analogy with our preceding considerations,
the related Green's functions
can be expressed in terms of
the expansions called matrix continued fractions
(interested readers may find more details, e.g., in \cite{MCF2}).

\subsection{The most elementary example.\\.}

In the majority of potential realistic applications of the
complex symmetric Hamiltonians
the matrix dimensions will be large so that
the experimentally testable predictions will have to be evaluated using
numerical (and, in particular, the present continued-fraction-based)
techniques.
Still, a useful methodical complement of
all of these calculations
will certainly be provided
by the non-numerically tractable models with small
matrix dimensions $N$ -- finite and,
in our present notation and for
the specific Bose-Hubbard-like models, equal to $K$.

In this manner we may return to $K=2$ and
reinterpret the coupling constant $\gamma$ in Eq.~(\ref{dopp2})
as a variable parameter $t$ mimicking time in
abbreviated $H^{(K)}(\gamma)=X(t)$.
Finally we may
recall the known
diagonalized form of the $K=2$ Hamiltonian,
 $$
 {X}(t) \to
 \mathfrak{\xi}(t)=\left[ \begin {array}{cc} \sqrt {1-{t}^{2}}&0\\
  \noalign{\medskip}0&-\sqrt {1-{t}^{2}}\end {array} \right]\,.
 $$
The elementary nature of the model enables us to
construct the explicit form of the product
 $$
 \mathbb{{X}}(t)={X}^\dagger(t)\,{X}(t)
 =\left[ \begin {array}{cc} {t}^{2}+1&2\,it\\
 \noalign{\medskip}-2\,it&{t}^{2}+1\end {array} \right]\,.
 $$
The singular values then follow from an easy
diagonalization of the matrix,
  $$
  \widetilde{\mathbb{{X}}}(t)
  =\left[ \begin {array}{cc} {t}^{2}+1+2\,t&0\\
  \noalign{\medskip}0&{t}^{2}+1-2\,t\end {array} \right]\,.
  $$
Interested readers may find more details in \cite{afew}.

\subsection{Further instructive few-state models.\\.}

A transition between the trivial and non-trivial
models emerges, at $N=K=4$, with Hamiltonian
 \be
 {X}(t)= \left[ \begin {array}{cccc}
  -3\,it&\sqrt {3}&0&0\\\noalign{\medskip}
 \sqrt {3}&-it&2&0\\\noalign{\medskip}0&2&it&\sqrt {3}
 \\
 \noalign{\medskip}0&0&\sqrt {3}&3\,it\end {array} \right]\,.
 \label{bufom}
 \ee
Its
diagonalized
partner with elements
 \be
 \mathfrak{\xi}_{11}(t)=\sqrt {1-t^2}=-\mathfrak{\xi}_{44}(t)\,,\ \ \
 \mathfrak{\xi}_{22}(t)=3\,\sqrt {1-t^2}=-\mathfrak{\xi}_{33}(t)\,
 \label{brac4}
 \ee
is still easy to deduce. In parallel, a more interesting matrix structure
is encountered in product
 $$
 \mathbb{{X}}(t)={X}^\dagger(t)\,{X}(t)=
  \left[ \begin {array}{cccc} 9\,{{t}}^{2}+3&2\,i{t}\sqrt {3}&2\,\sqrt {3}&0
\\\noalign{\medskip}-2\,i\sqrt {3}{t}&7+{{t}}^{2}&4\,i{t}&2\,\sqrt {3}
\\\noalign{\medskip}2\,\sqrt {3}&-4\,i{t}&7+{{t}}^{2}&2\,i{t}\sqrt {3}
\\\noalign{\medskip}0&2\,\sqrt {3}&-2\,i\sqrt {3}{t}&9\,{{t}}^{2}+3
\end {array} \right]
 $$
with
 eigenvalues
  $$
  \sigma_{\pm,\pm}(t)=
  5+5\,{{t}}^{2} \pm 2\,{t} \pm 4\,\sqrt {1-{t}-{{t}}^{3}+{{t}}^{4}}\,.
  $$
The analysis of the consequences of the further growth
of the matrix dimension can be found in the dedicated references \cite{4a5b}.
We only find it relevant to emphasize that the $N=6$ model
 \be
 {X}(t)= \left[ \begin {array}{cccccc} -5\,i{t}&\sqrt {5}&0&0&0&0\\
 \noalign{\medskip}\sqrt {5}&-3\,i{t}&2\,\sqrt {2}&0&0&0
\\\noalign{\medskip}0&2\,\sqrt {2}&-i{t}&3&0&0\\\noalign{\medskip}0&0&3&
i{t}&2\,\sqrt {2}&0\\\noalign{\medskip}0&0&0&2\,\sqrt {2}&3\,i{t}&\sqrt {5
}\\\noalign{\medskip}0&0&0&0&\sqrt {5}&5\,i{t}\end {array} \right]
 \label{bufor}
 \ee
still leads to the compact form of eigenvalues, yielding elements
$$
\mathfrak{\xi}_{11}(t)=5\,\sqrt {1-{t}^2}=-\mathfrak{\xi}_{66}(t)\,,\ \ \
\mathfrak{\xi}_{22}(t)=3\,\sqrt {1-{t}^2}=-\mathfrak{\xi}_{55}(t)\,,\ \ \
\mathfrak{\xi}_{33}(t)=\sqrt {1-{t}^2}=-\mathfrak{\xi}_{44}(t)\,
$$
of the
diagonalized Hamiltonian. In contrast, the formulae for the
related singular values, albeit explicit, already cease to be short and
easy to print.

 \noindent
\section{Matrix continued fractions.\\.}

Whenever one feels forced to
move to the use of numerical methods,
a return to their continued-fraction forms
can appear to be of a true relevance.
In particular, the above-mentioned
loss of transparency of
the singular-value formulae
in the
not yet too complicated
five-particle setting
indicates that
a return to the use of matrix
continued fractions might be truly well motivated.

\subsection{Standard (though, possibly, infinite-dimensional) matrix models.\\.}

For the purpose we may
start from a partitioned
factorization ansatz
 \be
 \mathbb{H}-\sigma
 = {\mathbb U}\,
 {\mathbb F}\,
 {\mathbb L}\,\,
 \label{alfinkit}
 \ee
in which ${\mathbb F} $ is block-diagonal, with two-by-two
elements
 $
 1/F_j  $.
 Once we follow
our preceding considerations we immediately arrive at
the two-by-two matrix continued fraction recurrences
 \be
 F_k=\frac{1}{A_k-\sigma-B_k F_{k+1}C_{k+1}}\,,\ \ \ \
 k=N, N-1,\ldots,2 ,1\,
 \label{macf}
 \ee
with
$F_{N+1}=0$ at  $N \leq \infty$.
Needless to add that the ultimate secular equation
 \be
 \det F_1^{-1}(E_n)=0\,
 \ee
remains user-friendly, and that also
a non-numerical construction
of the block-bidiagonal factor matrices
preserves the analogies, with
 \be
 {\mathbb U}=
  \left[ \begin {array}{cccc}
  I&-U_2&0
 &\ldots
   \\
     0&I&-U_3&\ddots
   \\
 0
 &0&I&\ddots
   \\
 \vdots&\ddots&\ddots&\ddots
    \\
 \end {array} \right]\,,
 \ \ \ \ \ \
 {\mathbb L}=
  \left[ \begin {array}{cccc}
  I&0&0
 &\ldots
   \\
     -V_2&I&0
 &\ldots
   \\
   0&-V_3&I&\ddots
   \\
 \vdots&\ddots
 &\ddots&\ddots
   \\
 \end {array} \right]\,
 \label{blowkit}
 \ee
leading to the
opportunity of constructing also wave functions (cf. \cite{MCF1}).

\subsection{Doubly-infinite matrix models.\\.}

Skipping the details (which can be found in \cite{afew})
let us now move to the climax of our present
considerations and mention the
possibility of encountering and using a doubly infinite partitioned
analogue of the conventional singular-value-related matrices (\ref{matrix}), viz.,
 \be
 {\mathbb H}=
  \left[ \begin {array}{ccc|c|ccc}
 \ddots&\ddots&\ddots&\vdots&&
    &\\
   \ddots&A_{-2}&B_{-2}&0&\vdots&
    &\\
   \ddots&C_{-1}&A_{-1}&B_{-1}&0&\ldots
    &\\
    \hline
  \ldots &0&C_0&A_0&B_0&0&\ldots
    \\
    \hline
    &\ldots
  &0& C_1&A_1&B_1&\ddots
   \\&& \vdots
 &0
 &C_2&A_2&\ddots
   \\&&
 &\vdots&\ddots&\ddots&\ddots
     \\
 \end {array} \right]\,.
 \label{bimatrix}
 \ee
It is worth emphasizing that there is a pair of
matrix continued fractions
 \be
 F_{-j}=\frac{1}{A_{-j}-\sigma-C_{-j}F_{-j-1}B_{-j-1}}\,,\ \ \ \
 j=M, M-1,\ldots,2 ,1
 \label{biucf}
 \ee
(with $F_{-M-1}=0$) and
 \be
 F_k=\frac{1}{A_k-\sigma-B_kF_{k+1}C_{k+1}}\,,\ \ \ \
 k=N, N-1,\ldots,2 ,1
 \label{bilcf}
 \ee
(where $F_{N+1}=0$) which
enters the game in such a case.

\begin{lemma}. \label{lemmautwo}
The Green's function associated with the Hermitian
quasi-Hamiltonian (\ref{bimatrix})
can be defined by formula ${G}(z)=\det F_0(z)$, with
 \be
 F_0(z)=\left [{A_{0}-z-C_{0}F_{-1}(z)B_{-1}-B_0F_{1}(z)
 C_{1}}\right ]^{-1}\,.
 \label{humacf}
 \ee
\end{lemma}
\begin{proof}
can be found in  \cite{MCF2}.
\end{proof}

\section{Summary.\\.}

Let us conclude that
along the lines which were sampled heere by the fairly realistic
Bose-Hubbard-like quantum models,
many
non-Hermitian tridiagonal Hamiltonians $H \neq H^\dagger$
may be found user-friendly. This seems to be true
even in the dynamical regime of
resonances
for which
the
singular values
were shown to be specified via  auxiliary, ``Hermitized'' Schr\"{o}dinger-like
equations.
In addition, the related
``Hermitized'' Green's functions
were given the two alternative compact and numerically efficient
matrix continued fraction forms.

%
%
%
%

\end{document}